\documentclass[journal,onecolumn,11pt]{IEEEtran}
\usepackage[cmex10]{amsmath}
\usepackage{amsthm}
\usepackage{amsmath}
\usepackage{amssymb}
\usepackage{amsfonts}
\usepackage{epic}
\usepackage{epsfig}
\usepackage{graphicx}
\usepackage[tight,footnotesize]{subfigure}
\usepackage{latexsym}
\usepackage{mathrsfs}
\usepackage{multirow}
\usepackage{newalg}
\usepackage[all]{xy}
\usepackage{cite}
\usepackage[usenames]{color}

\newtheorem{thm}{Theorem}
\newtheorem{lem}{Lemma}
\newtheorem{prop}{Proposition}

\theoremstyle{remark}
    \newtheorem{rem}{Remark}

\newcommand{\norm}[1]{\Vert #1 \Vert}

\newcommand{\normtwo}[1]{\Vert #1 \Vert_2}

\newcommand{\normone}[1]{\Vert #1 \Vert_1}

\newcommand{\gr}[1]{( #1 )}
\newcommand{\Gr}[1]{\big( #1 \big)}
\newcommand{\GR}[1]{\Big( #1 \Big)}
\newcommand{\GRg}[1]{\bigg( #1 \bigg)}

\newcommand{\AutoGroup}[1]{\left( #1 \right)}
\newcommand{\abs}[1]{\vert #1 \vert}

\newcommand{\set}[1]{\{ #1 \}}

\newcommand{\AutoSet}[1]{\left\{ #1 \right\}}

\newcommand{\AutoVect}[1]{\left[ #1 \right]}

\newcommand{\C}{\mathbb C}

\newcommand{\eps}{\varepsilon}

\newcommand{\bs}[1]{\boldsymbol{#1}}

\newcommand{\kth}{$k$th} %\newcommand{\kth}{$k^\textrm{th}$}

\newcommand{\cs}{compressed sensing}

\newcommand{\Wlog}{Without loss of generality}

\newcommand{\argmin}{\mathop{\mathrm{argmin}}}
\newcommand{\rank}{\mathop{\mathrm{rank}}}

\begin{document}
\bibliographystyle{IEEEtran}
%===========================================================================
%===================================================
%==========================

\title{General Deviants: An Analysis of Perturbations\\
in Compressed Sensing}
\author{\IEEEauthorblockN{Matthew A. Herman and Thomas Strohmer}
\thanks{The authors are with the Department of Mathematics, University of California,
Davis, CA 95616-8633, USA (\mbox{e-mail}: \texttt{\{mattyh,
strohmer\}@math.ucdavis.edu}).}
\thanks{This work was partially supported by NSF Grant No. DMS-0811169
and NSF VIGRE Grant No. DMS-0636297.}}

\markboth{Preprint}{Herman & Strohmer: General Deviants: An Analysis
of Perturbations in Compressed Sensing} % -- Revised \today}

\maketitle

\begin{abstract}
We analyze the Basis Pursuit recovery of signals with general
perturbations. Previous studies have only considered partially
perturbed observations $\bs{Ax} + \bs{e}$. Here, $\bs{x}$ is a
signal which we wish to recover, $\bs{A}$ is a full-rank matrix with
more columns than rows, and $\bs{e}$ is simple \emph{additive}
noise. Our model also incorporates perturbations~$\bs{E}$ to the
matrix~$\bs{A}$ which result in \emph{multiplicative} noise. This
completely perturbed framework extends the prior work of Cand\`es,
Romberg and Tao on stable signal recovery from incomplete and
inaccurate measurements. Our results show that, under suitable
conditions, the stability of the recovered signal is limited by the
noise level in the observation. Moreover, this accuracy is within a
constant multiple of the best-case reconstruction using the
technique of least squares. In the absence of additive noise
numerical simulations essentially confirm that this error is a
linear function of the relative perturbation.
\end{abstract}

%=================================================================
%=================================================================
\section{Introduction}
%=================================================================
%=================================================================
Employing the techniques of \cs\ (CS) to recover signals with a
sparse representation has enjoyed a great deal of attention over the
last 5--10 years. The initial studies considered an ideal
unperturbed scenario:
\begin{equation} \label{eq:Observation}
\bs{b} \;=\; \bs{Ax}.
\end{equation}
Here $\bs{b} \in\C^{m}$ is the observation vector, $\bs{A}
\in\C^{m\times n}$ is a full-rank measurement matrix or system model
(with $m\le n$), and $\bs{x} \in\C^{n}$ is the signal of interest
which has a sparse, or almost sparse, representation under some
fixed basis. More recently researchers have included an
\emph{additive} noise term~$\bs{e}$ into the received signal
\cite{CanRIP, CanRomTao_Noise, DonEladTemlyakov, Tropp_Relax}
creating a \emph{partially perturbed model}:
\begin{equation} \label{eq:Pert_observation}
\bs{\hat{b}} \;=\; \bs{Ax} + \bs{e}
\end{equation}
This type of noise typically models simple errors which are
uncorrelated with~$\bs{x}$.

As far as we can tell, practically no research has been done yet on
perturbations $\bs{E}$ to the matrix~$\bs{A}$.\footnote{A related
problem is considered in~\cite{Gribonval_UniteAverage} for greedy
algorithms rather than $\ell_1$-minimization, and in a multichannel
rather than a single channel setting; it mentions using different
matrices on the encoding and decoding sides, but its analysis is not
from an error or perturbation point of view.}${}^{,}$\footnote{At
the time of revising this manuscript we became aware of an earlier
study~\cite{BlumenDavies_CSSourseSeparation} which discusses the
error resulting from estimating the mixing matrix in source
separation problems. However, it only covers strictly sparse
signals, and its analysis is not as in   depth as presented in this
manuscript.} Our \textbf{\emph{completely perturbed model}} extends
(\ref{eq:Pert_observation}) by incorporating a perturbed sensing
matrix in the form of
$$\bs{\hat{A}} = \bs{A}+\bs{E}.$$
It is important to consider this kind of noise since it can account
for precision errors when applications call for physically
implementing the measurement matrix~$\bs{A}$ in a sensor. In other
CS scenarios, such as when~$\bs{A}$ represents a system
model,~$\bs{E}$ can absorb errors in assumptions made about the
transmission channel. This can be realized in
radar~\cite{HermanStrohmer_CSRadar}, remote
sensing~\cite{FannjiangYanStrohmer_CSRemote}, telecommunications,
source separation~\cite{Gribonval_UniteAverage,
BlumenDavies_CSSourseSeparation}, and countless other problems.
Further,~$\bs{E}$ can also model the distortions that result when
discretizing the domain of analog signals and systems; examples
include jitter error and choosing too coarse of a sampling period.

In general, these perturbations can be characterized as
\emph{multiplicative} noise, and are more difficult to analyze than
simple additive noise since they are correlated with the signal of
interest. To see this, simply substitute $\bs{A} = \bs{\hat{A}} -
\bs{E}$ in (\ref{eq:Pert_observation});\footnote{It essentially
makes no difference whether we account for the perturbation~$\bs{E}$
on the ``encoding side''~(\ref{eq:Pert_observation}), or on the
``decoding side''~(\ref{eq:Complete_Pert_BP}). The model used here
was chosen so as to agree with the conventions of classical
perturbation theory which we use in
Section~\ref{sect:Classical_Pert_Analysis}.} there will be an extra
noise term $\bs{Ex}$.

%It should be noted that parts of this analysis can be interpreted as
%a study on mixed operators akin to the study
%in~\cite{Gribonval_UniteAverage}. There are interesting lines of
%research which examine the conditions that permit two different
%matrices to be used in CS scenarios.

The rest of this section establishes certain assumptions and
notation necessary for our analysis.
Section~\ref{sect:Perturbation_Analysis} first gives a brief review
of previous work on the partially perturbed scenario in CS, and then
presents our main theoretical and numerical results on the
completely perturbed scenario. Section~\ref{sect:Proofs} provides
proofs of the theorems, and
Section~\ref{sect:Classical_Pert_Analysis} compares the CS solution
with classical least squares. Concluding remarks are given in
Section~\ref{sect:Conclusion}, and a brief discussion on different
kinds of perturbation $\bs{E}$ which we often encounter can be found
in the Appendix.

%-----------------------------------------------------------------
\subsection{Assumptions and Notation}
%-----------------------------------------------------------------
Throughout this paper we represent vectors and matrices with
boldface type. \Wlog, assume that the original data~$\bs{x}$ is a
$K$-sparse vector for some fixed $K$, or that it is compressible.
Vectors which are \emph{$K$-sparse} contain no more than $K$ nonzero
elements, and \emph{compressible} vectors are ones whose ordered
coefficients decay according to a power law (i.e.,
$\abs{\bs{x}}_{\gr{k}}\le C_p \:\! k^{-p}$, where
$\abs{\bs{x}}_{\gr{k}}$ is the \kth\ largest element of~$\bs{x}$,
$p\ge1$, and~$C_p$ is a constant which depends only on~$p$). Let
vector $\bs{x}_K\in\C^n$ be the best $K$-term approximation to
$\bs{x}$, i.e., it contains the $K$ largest coefficients of~$\bs{x}$
with the rest set to zero. We occasionally refer to this vector as
the ``head'' of~$\bs{x}$. Note that if~$\bs{x}$ is $K$-sparse, then
$\bs{x}=\bs{x}_K$. With a slight abuse of notation denote
$\bs{x}_{K^c} = \bs{x}-\bs{x}_K$ as the ``tail'' of $\bs{x}$.

The symbols $\sigma_{\max}\gr{\bs{Y}}$, $\sigma_{\min}\gr{\bs{Y}}$,
and $\normtwo{\bs{Y}}$ respectively denote the usual maximum,
minimum nonzero singular values, and spectral norm of a matrix
$\bs{Y}$. Our analysis will require examination of
\emph{submatrices} consisting of an arbitrary collection of $K$
columns. We use the superscript~$\gr{K}$ to represent extremal
values of the above spectral measures. For instance,
$\sigma_{\max}^{\gr{K}}\gr{\bs{Y}}$ denotes the largest singular
value taken over all \emph{$K$-column submatrices} of~$\bs{Y}$.
Similar definitions apply to $\normtwo{\bs{Y}}^{\gr{K}}$ and
${\rank}^{\gr{K}}\gr{\bs{Y}}$, while
$\sigma_{\min}^{\gr{K}}\gr{\bs{Y}}$ is the \emph{smallest} nonzero
singular value over all $K$-column submatrices of~$\bs{Y}$. With
these, the perturbations $\bs{E}$ and $\bs{e}$
can be quantified with the following relative bounds\\
\begin{equation} \label{eq:Pert_relative}
\frac{\normtwo{\bs{E}}} {\normtwo{\bs{A}}} \,\le\, \eps_{\bs{A}},
\quad \frac{\normtwo{\bs{E}}^{\gr{K}}} {\normtwo{\bs{A}}^{\gr{K}}}
\,\le\, \eps_{\bs{A}}^{\gr{K}}, \quad
\frac{\normtwo{\bs{e}}}{\normtwo{\bs{b}}} \,\le\, \eps_{\bs{b}},\\
\end{equation}
where
$\normtwo{\bs{A}},\normtwo{\bs{A}}^{\gr{K}},\normtwo{\bs{b}}\neq0$.
In real-world applications we often do not know the exact nature of
$\bs{E}$ and $\bs{e}$ and instead are forced to estimate their
relative upper bounds. This is the point of view taken throughout
most of this treatise. In this study we are only interested in the
case where $\eps_{\bs{A}}, \eps_{\bs{A}}^{\gr{K}}, \eps_{\bs{b}} <
1$.\\

%=================================================================
%=================================================================
\section{CS $\ell_1$ Perturbation Analysis} \label{sect:Perturbation_Analysis}
%=================================================================
%=================================================================
%-----------------------------------------------------------------
\subsection{Previous Work}
%-----------------------------------------------------------------
In the \emph{partially perturbed scenario} (i.e.,~$\bs{E} = \bs{0}$)
we are concerned with solving the \emph{Basis Pursuit}~(BP)
problem~\cite{ChenDonSaun}:
\begin{equation} \label{eq:Candes_Pert_BP}
\bs{z^\star} = \argmin_{\bs{\hat{z}}} {\norm{\bs{\hat{z}}}_1} \;\;\,
\mathrm{s.t.} \;\;\, \normtwo{\bs{A\hat{z}} - \bs{\hat{b}}} \,\leq\,
\eps'
\end{equation}
for some $\eps'\ge0$.\footnote{Throughout this paper \emph{absolute}
errors are denoted with a prime. In contrast, \emph{relative}
perturbations, such as in~(\ref{eq:Pert_relative}), are not primed.}

The \emph{restricted isometry property} (RIP)
\cite{CanTao_DecLinProg} for any matrix~$\bs{A}\in\C^{m\times n}$
defines, for each integer $K=1,2,\ldots$, the \emph{restricted
isometry constant} (RIC) $\delta_K$, which is the smallest
nonnegative number such that
\begin{equation} \label{eq:RIP}
\gr{1-\delta_K}\normtwo{\bs{x}}^2 \:\le\: \normtwo{\bs{Ax}}^2
\:\le\: \gr{1+\delta_K}\normtwo{\bs{x}}^2
\end{equation}
holds for any $K$-sparse vector $\bs{x}$. In the context of the RIC,
we observe that $\normtwo{\bs{A}}^{\gr{K}} =
\sigma_{\max}^{\gr{K}}\gr{\bs{A}} \le \sqrt{1+\delta_K}$, and
$\sigma_{\min}^{\gr{K}}\gr{\bs{A}} \ge \sqrt{1-\delta_K}$.

Assuming $\delta_{2K} < \sqrt{2}-1$ and $\normtwo{\bs{e}}\le\eps'$,
Cand\`es has shown (\!\!\cite{CanRIP}, Thm.~1.2) that the solution
to~(\ref{eq:Candes_Pert_BP}) obeys
\begin{equation} \label{eq:Candes_result}
\normtwo{\bs{z^\star} - \bs{x}} \,\le\,
C_0\:\!K^{-1/2}\normone{\bs{x}-\bs{x}_K} + C_1\:\!\eps'
\end{equation}
for some constants $C_0, C_1\ge0$ which are reasonably well-behaved
and can be calculated explicitly.\\

%-----------------------------------------------------------------
\subsection{Incorporating nontrivial perturbation $\bs{E}$}
%-----------------------------------------------------------------
Now assume the \textbf{\emph{completely perturbed}} situation with
$\bs{E},\bs{e}\ne\bs{0}$. In this case the BP problem
of~(\ref{eq:Candes_Pert_BP}) can be generalized to include a
different decoding matrix~$\bs{\hat{A}}$:
\begin{equation} \label{eq:Complete_Pert_BP}
\bs{z^\star} = \argmin_{\bs{\hat{z}}} {\norm{\bs{\hat{z}}}_1} \;\;\,
\mathrm{s.t.} \;\;\, \normtwo{\bs{\hat{A}\hat{z}} - \bs{\hat{b}}}
\,\leq\, \eps'_{\bs{A},K,\bs{b}}
\end{equation}
for some $\eps'_{\bs{A},K,\bs{b}}\ge0$. The following two theorems
summarize our results.

\begin{thm}[RIP for $\bs{\hat{A}}$] \label{thm:RIP_perturbed}
Fix $K=1,2,\ldots$. Given the RIC~$\delta_K$ associated with
matrix~$\bs{A}$ in~(\ref{eq:RIP}) and the relative
perturbation~$\eps_{\bs{A}}^{\gr{K}}$ associated with (possibly
unknown) matrix~$\bs{E}$ in~(\ref{eq:Pert_relative}), fix the
constant
\begin{equation} \label{eq:pert_RIC_bound}
\hat{\delta}_{K\!,\:\!\max} \,:=\, \Gr{1+\delta_K}
\GR{1+\eps_{\bs{A}}^{\gr{K}}}^2- 1.
\end{equation}
Then the RIC $\hat{\delta}_K$ for matrix $\bs{\hat{A}} = \bs{A} +
\bs{E}$ is the smallest nonnegative number such that
\begin{equation} \label{eq:Pert_RIP}
\gr{1-\hat{\delta}_K}\normtwo{\bs{x}}^2 \:\le\:
\normtwo{\bs{\hat{A}x}}^2 \:\le\:
\gr{1+\hat{\delta}_K}\normtwo{\bs{x}}^2
\end{equation}
holds for any $K$-sparse vector $\bs{x}$ where $\hat{\delta}_K \le
\hat{\delta}_{K\!,\:\!\max}$.
\end{thm}

\begin{rem}
Properly interpreting Theorem~\ref{thm:Completely_perturbed_BP} is
important. It is assumed that the only information known about
matrix~$\bs{E}$ is its \emph{worst-case} relative
perturbation~$\eps_{\bs{A}}^{\gr{K}}$, and therefore the bound of
$\hat{\delta}_{K\!,\:\!\max}$ in~(\ref{eq:pert_RIC_bound})
represents a \emph{worst-case deviation} of~$\hat{\delta}_K$. Notice
for a given~$\eps_{\bs{A}}^{\gr{K}}$ that there are infinitely
many~$\bs{E}$ which satisfy it. In fact, it is possible to construct
nonzero perturbations which result in $\hat{\delta}_K = \delta_K$!
For example, suppose $\bs{\hat{A}} = \bs{AU}$ for some unitary
matrix $\bs{U}\ne\bs{I}$ where $\bs{I}$ is the identity matrix.
Clearly here $\bs{E} = \bs{A}\gr{\bs{U}-\bs{I}} \ne \bs{0}$ and yet
since $\bs{U}$ is unitary we have $\hat{\delta}_K = \delta_K$. In
this case  using~$\eps_{\bs{A}}^{\gr{K}}$ to
calculate~$\hat{\delta}_{K\!,\:\!\max}$ could be a gross upper bound
for~$\hat{\delta}_K$. If more information on~$\bs{E}$ is
known,\footnote{See the appendix for more discussion on the
different forms of perturbation~$\bs{E}$ which we are likely to
encounter.} then much tighter bounds on~$\hat{\delta}_K$ can be
determined.
\end{rem}

\begin{rem}
The flavor of the RIP is defined with respect to the square of the
operator norm. That is, $\gr{1-\delta_K}$ and $\gr{1+\delta_K}$ are
measures of the \textbf{\emph{square}} of the minimum and maximum
singular values of $K$-column submatrices of~$\bs{A}$, and similarly
for $\bs{\hat{A}}$. In keeping with the convention of classical
perturbation theory however, we defined $\eps_{\bs{A}}^{\gr{K}}$
in~(\ref{eq:Pert_relative}) just in terms of the operator norm (not
its square). Therefore, the quadratic dependence
of~$\hat{\delta}_{K\!,\:\!\max}$ on~$\eps_{\bs{A}}^{\gr{K}}$
in~(\ref{eq:pert_RIC_bound}) makes sense. Moreover, in discussing
the spectrum of $K$-column submatrices of~$\bs{\hat{A}}$, we see
that it is really a \emph{linear function}
of~$\eps_{\bs{A}}^{\gr{K}}$.
\end{rem}

Before introducing the next theorem let us define the following
constants due to matrix $\bs{A}$
\begin{equation} \label{eq:condition_numbers}
\kappa_{\bs{A}}^{\gr{K}} :=
\frac{\sqrt{1+\delta_K}}{\sqrt{1-\delta_K}}, \qquad \alpha_{\bs{A}}
:= \frac{\normtwo{\bs{A}}}{\sqrt{1-\delta_K}}.
\end{equation}
The first quantity bounds the ratio of the extremal singular values
of all $K$-column submatrices of~$\bs{A}$
$$\frac{\sigma_{\max}^{\gr{K}}\gr{\bs{A}}}
{\sigma_{\min}^{\gr{K}}\gr{\bs{A}}} \,\le\,
\kappa_{\bs{A}}^{\gr{K}}.$$
Actually, for very small~$\delta_K$ we
have $\kappa_{\bs{A}}^{\gr{K}}\approx1$, which implies that every
$K$-column submatrix forms an approximately orthonormal set.

Also introduce the ratios
\begin{equation} \label{eq:ratios}
r_K := \frac{\normtwo{\bs{x}_{K^c}}}{\normtwo{\bs{x}_K}}, \qquad s_K
:= \frac{\normone{\bs{x}_{K^c}}}{\normtwo{\bs{x}_K}}
\end{equation}
which quantify the weight of a signal's tail relative to its head.
When~$\bs{x}$ is $K$-sparse we have $\bs{x}_{K^c} = \bs{0}$, and so
$r_K = s_K = 0$. If $\bs{x}$ is compressible, then these values are
a function of the power $p$ (i.e., the rate at which the
coefficients decay), and the cardinality $K$ of the group of its
largest entries. For reasonable values of $p$ and $K$, we expect
that $r_K, s_K \ll 1$.

\begin{thm}[Stability from completely perturbed observation]
\label{thm:Completely_perturbed_BP} Fix the relative perturbations
$\eps_{\bs{A}}$, $\eps_{\bs{A}}^{\gr{K}}$, $\eps_{\bs{A}}^{\gr{2K}}$
and~$\eps_{\bs{b}}$ in~(\ref{eq:Pert_relative}). Assume the RIC for
matrix $\bs{A}$ satisfies$^{ }$\footnote{Note for $\delta_{2K} \ge
0$, (\ref{cond:Main_thm_constraint_1}) requires that
$\eps_{\bs{A}}^{(2K)}< \sqrt[4]{2}-1$.}
\begin{equation} \label{cond:Main_thm_constraint_1}
\delta_{2K} \;<\; \frac{\sqrt{2}}
{\GR{1+\eps_{\bs{A}}^{\gr{2K}}}^{2}} \,-\, 1,
\end{equation} and that general signal $\bs{x}$ satisfies
\begin{equation} \label{cond:Main_thm_constraint_2}
r_K + \frac{s_K}{\sqrt{K}} \;<\; \frac{1}{\kappa_{\bs{A}}^{\gr{K}}}.
\end{equation}
Set the total noise parameter
\begin{equation} \label{eq:BP_absolute_error_constraint}
\eps'_{\bs{A},K,\bs{b}} := \GRg{
\frac{\eps_{\bs{A}}^{\gr{K}}\kappa_{\bs{A}}^{\gr{K}} +\,
\eps_{\bs{A}}\alpha_{\bs{A}}r_K}
{1-\kappa_{\bs{A}}^{\gr{K}}\!\Gr{r_K + {s_K}/{\sqrt{K}}}} \,+\,
\eps_{\bs{b}}} \normtwo{\bs{b}}.
\end{equation}
Then the solution of the BP problem
(\ref{eq:Complete_Pert_BP}) obeys \\
\begin{equation} \label{eq:Pert_error_soln_CS_l1}
{\normtwo{\bs{z^\star} - \bs{x}}} \;\le\;
\frac{C_0}{\sqrt{K}}\:\!\normone{\bs{x}-\bs{x}_K} \,+\,
C_1\:\!\eps'_{\bs{A},K,\bs{b}}, \\
\end{equation}
where
\begin{equation} \label{eq:Pert_Constant0_BP}
C_0 =
\frac{2\AutoGroup{1\;\!+\;\!\Gr{\sqrt{2}-1}\AutoVect{\Gr{1+\delta_{2K}}
\AutoGroup{1+\eps_{\bs{A}}^{\gr{2K}}}^2 - 1}}}
{1\;\!-\;\!\Gr{\sqrt{2}+1}\AutoVect{\Gr{1+\delta_{2K}}\AutoGroup{1+\eps_{\bs{A}}^{\gr{2K}}}^2
- 1}},
\end{equation}
\begin{equation} \label{eq:Pert_Constant1_BP}
C_1 =
\frac{4\sqrt{1+\delta_{2K}}\AutoGroup{1+\eps_{\bs{A}}^{\gr{2K}}}}
{1\;\!-\;\!\Gr{\sqrt{2}+1}\AutoVect{\Gr{1+\delta_{2K}}\AutoGroup{1+\eps_{\bs{A}}^{\gr{2K}}}^2
- 1}}.
\end{equation}
\end{thm}

\begin{rem}
Theorem~\ref{thm:Completely_perturbed_BP} generalizes Cand\`es'
results in~\cite{CanRIP}. Indeed, if matrix $\bs{A}$ is unperturbed,
then $\bs{E} = \bs{0}$ and $\eps_{\bs{A}} = \eps_{\bs{A}}^{\gr{K}} =
0$. It follows that $\hat{\delta}_{K} = \delta_{K}$
in~(\ref{eq:pert_RIC_bound}), and the RIPs for $\bs{A}$ and
$\bs{\hat{A}}$ coincide. Moreover,
assumption~(\ref{cond:Main_thm_constraint_1}) in
Theorem~\ref{thm:Completely_perturbed_BP} reduces to
$\delta_{K}<\sqrt{2}-1$, and the total perturbation
(see~(\ref{eq:Total_pert})) collapses to $\normtwo{\bs{e}} \le
\eps'_{\bs{b}} := \eps_{\bs{b}}\normtwo{\bs{b}}$ (so that
assumption~(\ref{cond:Main_thm_constraint_2}) is no longer
necessary); both of these are identical to Cand\`es' assumptions in
(\ref{eq:Candes_result}). Finally, the constants~$C_0, C_1$
in~(\ref{eq:Pert_Constant0_BP}) and~(\ref{eq:Pert_Constant1_BP})
reduce to the same as outlined in the proof of~\cite{CanRIP}.
\end{rem}

The assumption in~(\ref{cond:Main_thm_constraint_2}) demands more
discussion. Observe that the left-hand side (LHS) is solely a
function of the signal~$\bs{x}$, while the right-hand side (RHS) is
just a function of the matrix~$\bs{A}$. For reasonably compressible
signals, it is often the case that the LHS is on the order of
$10^{-2}$ or $10^{-3}$. At the same time, the RHS is always of order
$10^0$ due to assumption~(\ref{cond:Main_thm_constraint_1}).
Therefore, there should be a sufficient gap to ensure that
assumption~(\ref{cond:Main_thm_constraint_2}) holds. Clearly this
condition is automatically satisfied whenever~$\bs{x}$ is strictly
$K$-sparse.

In fact, more can be said about
Theorem~\ref{thm:Completely_perturbed_BP} for the case of a
$K$-sparse input. Notice then that the terms related
to~$\bs{x}_{K^c}$ in~(\ref{eq:BP_absolute_error_constraint})
and~(\ref{eq:Pert_error_soln_CS_l1}) disappear, and the accuracy of
the solution becomes
\begin{equation*}
{\normtwo{\bs{z^\star} - \bs{x}}} \;\le\; C_1 \AutoGroup{
\kappa_{\bs{A}}^{\gr{K}}\;\!\eps_{\bs{A}}^{\gr{K}} \,+\,
\eps_{\bs{b}}} \normtwo{\bs{b}}.
\end{equation*}
This form of the stability of the BP solution is helpful since it
highlights the effect of the perturbation~$\bs{E}$ on the~$K$ most
important elements of~$\bs{x}$, as well as the influence of the
additive noise $\bs{e}$. Clearly in the absence of any perturbation,
a $K$-sparse signal can be perfectly recovered by BP.

It is also interesting to examine the spectral effects due to the
first assumption of Theorem~\ref{thm:Completely_perturbed_BP}.
Namely, we want to be assured that the maximum rank of submatrices
of $\bs{A}$ is unaltered by the perturbation~$\bs{E}$.
\begin{lem} \label{lem:Spectral_consequence}
Assume condition (\ref{cond:Main_thm_constraint_1}) of
Theorem~\ref{thm:Completely_perturbed_BP} holds. Then for any
$k\le2K$
\begin{equation} \label{eq:Spectral_consequence}
\sigma_{\max}^{\gr{k}}\gr{\bs{E}} \;<\;
\sigma_{\min}^{\gr{k}}\gr{\bs{A}},
\end{equation}
and therefore
$${\rank}^{\gr{k}}\gr{\bs{\hat{A}}}
\;=\; {\rank}^{\gr{k}}\gr{\bs{A}}.$$
\end{lem}
\noindent We apply this fact in the least squares analysis of
Section~\ref{sect:Classical_Pert_Analysis}.\\[-5pt]

The utility of Theorems~\ref{thm:RIP_perturbed}
and~\ref{thm:Completely_perturbed_BP} can be understood with two
simple numerical examples. Suppose that matrix~$\bs{A}$
in~(\ref{eq:Pert_observation}) represents a system that a signal
passes through which in reality has an RIC of $\delta_{2K} = 0.100$.
Assume however, that when modeling this system we introduce a
worst-case relative error of $\eps_{\bs{A}}^{\gr{2K}} = 5\%$ so that
we think that the system behaves as $\bs{\hat{A}} = \bs{A}+\bs{E}$.
From~(\ref{eq:pert_RIC_bound}) we can verify that
matrix~$\bs{\hat{A}}$ has an RIC $\hat{\delta}_{2K\!,\:\!\max} =
0.213$ which satisfies~(\ref{cond:Main_thm_constraint_1}). Thus,
if~(\ref{cond:Main_thm_constraint_2}) is also satisfied, then
Theorem~\ref{thm:Completely_perturbed_BP} guarantees that the BP
solution will have accuracy given
in~(\ref{eq:Pert_error_soln_CS_l1}) with $C_0 = 4.47$ and $C_1 =
9.06$. Note from~(\ref{eq:Pert_Constant0_BP})
and~(\ref{eq:Pert_Constant1_BP}) we see that if there had been no
perturbation, then $C_0 = 2.75$ and $C_1 = 5.53$.

Consider now a different example. Suppose instead that
$\delta_{2K}=0.200$ with $\eps_{\bs{A}}^{\gr{2K}} = 1\%$. Then
$\hat{\delta}_{2K\!,\:\!\max} = 0.224$, $C_0 = 4.76$ and $C_1 =
9.64$. Here, if $\bs{A}$ was unperturbed, then we would have had
$C_0 = 4.19$ and $C_1 = 8.47$.

These numerical examples show how the stability constants $C_0$ and
$C_1$ of the BP solution get worse with perturbations to~$\bs{A}$.
It must be stressed however, that they represent worst-case
instances. It is well-known in the CS community that better
performance is normally achieved in practice.

%=================================================================
%=================================================================
\subsection{Numerical Simulations}
%=================================================================
%=================================================================
Numerical simulations were conducted in {\sc Matlab} as follows. In
each trial a new matrix~$\bs{A}$ of size $128\times512$ was randomly
generated with normally distributed entries
$\mathcal{N}\gr{0,\sigma^2}$ where $\sigma^2 = 1/128$ (so that the
expected $\ell_2$-norm of each column was unity), and the spectral
norm of~$\bs{A}$ was calculated. Next, for each relative
perturbation $\eps_{\bs{A}}= 0, 0.01, 0.05, 0.1$ a different
perturbation matrix $\bs{E}$ with normally distributed entries was
generated, and then scaled so that $\normtwo{\bs{E}} =
\eps_{\bs{A}}\cdot\normtwo{\bs{A}}$.\footnote{We used
$\eps_{\bs{A}}$ in these simulations since calculating
$\eps_{\bs{A}}^{\gr{K}}$ explicitly is extremely difficult. Notice
that $\eps_{\bs{A}}\approx\eps_{\bs{A}}^{\gr{K}}$ for all $K$ with
high probability since both $\bs{A}, \bs{E}$ are random Gaussian
matrices.} A random vector $\bs{x}$ of sparsity $K=1,\ldots,64$ was
then randomly generated with nonzero entries uniformly distributed
$\mathcal{N}\gr{0,1}$, and $\bs{\hat{b}} = \bs{Ax}$
in~(\ref{eq:Pert_observation}) was created (note, we set
$\bs{e}=\bs{0}$ so as to focus on the effect of perturbation
$\bs{E}$). Finally, given $\bs{\hat{b}}$ and the $\bs{\hat{A}} =
\bs{A}+\bs{E}$ associated with each $\eps_{\bs{A}}$, the BP
program~(\ref{eq:Complete_Pert_BP}) was implemented with {\texttt
cvx} software~\cite{Boyd_cvx} and the relative error
$\normtwo{\bs{z^\star} - \bs{x}}/\normtwo{\bs{x}}$ was recorded. One
hundred trials were performed for each value of~$K$.

Figure~\ref{fig:MatrixPerturb_K64_m128_n512_aver100} shows the
relative error averaged over the $100$ trials as a function of~$K$
for each~$\eps_{\bs{A}}$. As a reference, the ideal, noise-free case
can be seen for $\eps_{\bs{A}}=0$. Now fix a particular value of
$K\le30$ and compare the relative error for the three nonzero values
of~$\eps_{\bs{A}}$. It is clear that the error scales roughly
linearly with~$\eps_{\bs{A}}$. For example, when $K=10$ the relative
errors corresponding to $\eps_{\bs{A}} = 0.01, 0.05, 0.1$
respectively are $9.7\times10^{-3}, 4.9\times10^{-2},
9.7\times10^{-2}$. We see here that the relative errors for
$\eps_{\bs{A}} = 0.05$ and $0.1$ are approximately five and ten
times the the relative error associated with $\eps_{\bs{A}} = 0.01$.
Therefore, this empirical study essentially confirms the conclusion
of Theorem~\ref{thm:Completely_perturbed_BP}: the stability of the
BP
solution scales linearly with $\eps_{\bs{A}}^{\gr{K}}$.\\
% ========================
% Figure Included Here!
% ========================
\begin{figure} [!t]
\begin{center}
\includegraphics[scale=0.55, trim=0 0 0 0mm, clip]{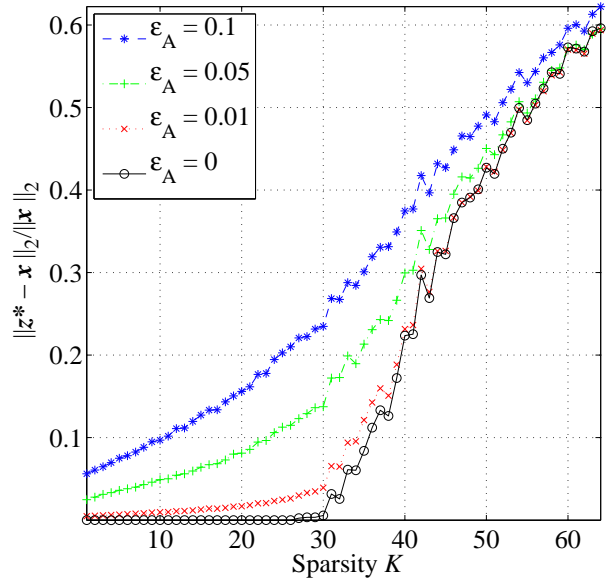}
\caption{Average (100 trials) relative error of BP solution
$\bs{z^\star}$ with respect to $K$-sparse $\bs{x}$ vs. Sparsity $K$
for different relative perturbations~$\eps_{\bs{A}}$ of $\bs{A}$.
Here $\bs{A}, \bs{E}$ are both $128\times512$ random matrices with
i.i.d. Gaussian entries and $\eps_{\bs{b}}=0$.
\label{fig:MatrixPerturb_K64_m128_n512_aver100}}
\end{center}
\end{figure}

Note that improved performance in theory and in simulation can be
achieved if BP is used solely to determine the support of the
solution. Then we can use least squares to better approximate the
coefficients on this support. This is similar to the the best-case,
oracle least squares solution discussed in
Section~\ref{sect:Classical_Pert_Analysis}. However, this method of
recovery was not pursued in the present analysis.

%=================================================================
%=================================================================
\section{Proofs} \label{sect:Proofs}
%=================================================================
%=================================================================
%-----------------------------------------------------------------
\subsection{Proof of Theorem~\ref{thm:RIP_perturbed}} \label{sect:Perturbed_RIP}
%-----------------------------------------------------------------
Recall that we are tasked with determining the
maximum~$\hat{\delta}_K$ given~$\delta_K$
and~$\eps_{\bs{A}}^{\gr{K}}$. Temporarily define~$l_K$ and~$u_K$ as
the \emph{smallest nonnegative numbers} such that
\begin{equation} \label{eq:Pert_RIP_Temporary}
\gr{1-l_K}\normtwo{\bs{x}}^2 \,\le\, \normtwo{\bs{\hat{A}x}}^2
\,\le\, \gr{1+u_K}\normtwo{\bs{x}}^2
\end{equation}
holds for any $K$-sparse vector $\bs{x}$. From the triangle
inequality,~(\ref{eq:RIP}) and (\ref{eq:Pert_relative}) we have
\begin{eqnarray}
\normtwo{\bs{\hat{A}x}}^2
&\le& \Gr{\normtwo{\bs{Ax}} \,+\, \normtwo{\bs{Ex}}}^2 \label{eq:Perturbed_RIP_Temp1}\\
&\le& \GR{\sqrt{1 + \delta_K} \;\!+\;\!
\normtwo{\bs{E}}^{\gr{K}}}^2 \normtwo{\bs{x}}^2 \label{eq:Perturbed_RIP_Temp2} \\
&\le& \gr{1 + \delta_K}\GR{1 \;\!+\;\! \eps_{\bs{A}}^{\gr{K}}}^2
\normtwo{\bs{x}}^2. \label{eq:Perturbed_RIP_Temp3}
\end{eqnarray}
In comparing the RHS of~(\ref{eq:Pert_RIP_Temporary}) and
(\ref{eq:Perturbed_RIP_Temp3}), it must be that
\begin{equation*}
\gr{1+u_K} \:\le\: \gr{1 + \delta_K}\GR{1 \;\!+\;\!
\eps_{\bs{A}}^{\gr{K}}}^2
\end{equation*}
as demanded by the definition of the $u_K$. Moreover, this
inequality is sharp for the following reasons:
\begin{itemize}
\item Equality occurs in~(\ref{eq:Perturbed_RIP_Temp1}) whenever $\bs{E}$ is a
positive, real-valued multiple of $\bs{A}$.
\item The inequality in~(\ref{eq:Perturbed_RIP_Temp2}) inherits the
sharpness of the upper bound of the RIP for matrix $\bs{A}$ in
(\ref{eq:RIP}).
\item Equality occurs in~(\ref{eq:Perturbed_RIP_Temp3}) since, in this hypothetical
case, we assume that $\bs{E} = \beta\bs{A}$ for some $0<\beta<1$.
Therefore, the relative perturbation~$\eps_{\bs{A}}^{\gr{K}}$
in~(\ref{eq:Pert_relative}) no longer represents a worst-case
deviation (i.e., the ratio $\frac{\normtwo{\bs{E}}^{\gr{K}}}
{\normtwo{\bs{A}}^{\gr{K}}} = \beta =: \eps_{\bs{A}}^{\gr{K}}$).
\end{itemize}
Since the triangle inequality constitutes a \emph{least-upper
bound}, and since we attain this bound, then
\begin{equation*}
u_K \::=\: \gr{1 + \delta_K}\GR{1 \;\!+\;\!
\eps_{\bs{A}}^{\gr{K}}}^2 -1
\end{equation*}
satisfies the definition of $u_K$.

Now the LHS of~(\ref{eq:Pert_RIP_Temporary}) is obtained in much the
same way using the ``reverse'' triangle inequality with similar
arguments (in particular, assume $-1<\beta<0$ and
$\eps_{\bs{A}}^{\gr{K}} := \abs{\beta}$). Thus
\begin{equation*}
l_K \::=\: 1 - \gr{1-\delta_K}\GR{1 \;\!-\;\!
\eps_{\bs{A}}^{\gr{K}}}^2.
\end{equation*}
Next, we need to make the bounds of~(\ref{eq:Pert_RIP_Temporary})
symmetric. Notice that $\gr{1-u_K} \,\le\, \gr{1-l_K}$ and
$\gr{1+l_K} \,\le\, \gr{1+u_K}$. Therefore, given $\delta_K$ and
$\eps_{\bs{A}}^{\gr{K}}$, we choose
$$\hat{\delta}_{K\!,\:\!\max} \,:=\, u_K$$
as the smallest nonnegative constant which
makes~(\ref{eq:Pert_RIP_Temporary}) symmetric. Finally, it is clear
that the actual RIC~$\hat{\delta}_K$ %associated with
for~$\bs{\hat{A}}$ obeys $\hat{\delta}_K \le
\hat{\delta}_{K\!,\:\!\max}$. Hence,~(\ref{eq:Pert_RIP}) follows
immediately. \hfill $\blacksquare$

%-----------------------------------------------------------------
\subsection{Bounding the perturbed observation} \label{sect:Bounding_Perturbed_Observation}
%-----------------------------------------------------------------
Before proceeding to the proof of
Theorem~\ref{thm:Completely_perturbed_BP} we need several important
facts. First we generalize a lemma in~\cite{OneSketchForAll} about
the image of an arbitrary signal.

\begin{prop}[\!\!\cite{OneSketchForAll}, Lemma~29] \label{prop:OneSketchForAll}
Assume that matrix $\bs{A}$ satisfies the upper bound of the RIP
in~(\ref{eq:RIP}). Then for every signal $\bs{x}$ we have
$$\normtwo{\bs{Ax}} \le \sqrt{1+\delta_K}\GR{\normtwo{\bs{x}} +
\frac{1}{\sqrt{K}}\normone{\bs{x}}}.$$
\end{prop}

Now we can establish sufficient conditions for the lower bound in
terms of the head and tail of $\bs{x}$ and the RIC of~$\bs{A}$.
\begin{lem} \label{lem:Bound_Image_Tail}
Assume condition~(\ref{cond:Main_thm_constraint_2}) in
Theorem~\ref{thm:Completely_perturbed_BP}. Then for general signal
$\bs{x}$, its image under $\bs{A}$ can be bounded below by the
positive quantity
$$\normtwo{\bs{Ax}} \ge \sqrt{1-\delta_K}\GR{\normtwo{\bs{x}_K}
\,-\, \kappa_{\bs{A}}^{\gr{K}}\GR{\normtwo{\bs{x}_{K^c}} +
\frac{\normone{\bs{x}_{K^c}}}{\sqrt{K}}}}.$$
\end{lem}

\begin{proof}
Apply Proposition~\ref{prop:OneSketchForAll} to the tail of
$\bs{x}$. Then
\begin{eqnarray}\normtwo{\bs{Ax}} &\!\!\!\ge&
\!\!\!\normtwo{\bs{A\:\!x}_K}-\normtwo{\bs{A\:\!x}_{K^c}} \nonumber\\
&\!\!\!\ge& \!\!\!\!\sqrt{1-\delta_K}\normtwo{\bs{x}_K} \:\!\!-
\!\sqrt{1+\delta_K}\GR{\normtwo{\bs{x}_{K^c}} +
\frac{\normone{\bs{x}_{K^c}}}{\sqrt{K}}} \nonumber\\
&\!\!\!=& \!\!\!\sqrt{1-\delta_K}\GR{1 -
\kappa_{\bs{A}}^{\gr{K}}\GR{r_K +
\frac{s_K}{\sqrt{K}}}}\normtwo{\bs{x}_K} \nonumber\\
&\!\!\!>& \!\!\!0 \nonumber
\end{eqnarray}
on account of~(\ref{cond:Main_thm_constraint_2}).
\end{proof}

\vspace{10pt} We still need some sense of the size of the total
perturbation incurred by $\bs{E}$ and $\bs{e}$. We do not know
\emph{a priori} the exact values of $\bs{E}$, $\bs{x}$, or $\bs{e}$.
But we can find an upper bound in terms of the relative
perturbations in~(\ref{eq:Pert_relative}). The main goal in the
following lemma is to remove the total perturbation's dependence on
the input $\bs{x}$.

\begin{lem}[Total perturbation bound] \label{lem:Total_observed_perturbation_bound}
Assume condition~(\ref{cond:Main_thm_constraint_2}) in
Theorem~\ref{thm:Completely_perturbed_BP} and set$\,$\footnote{Note
that the results in this paper can easily be expressed in terms of
the perturbed observation by replacing $\normtwo{\bs{b}} \le
{\normtwo{\bs{\hat{b}}}}\gr{1 - \eps_{\bs{b}}}^{-1}.$ This can be
useful in practice since one normally only has access
to~$\bs{\hat{b}}$.}
$$\eps'_{\bs{A},K,\bs{b}} := \GRg{
\frac{\eps_{\bs{A}}^{\gr{K}}\kappa_{\bs{A}}^{\gr{K}} +\,
\eps_{\bs{A}}\alpha_{\bs{A}}r_K}
{1-\kappa_{\bs{A}}^{\gr{K}}\!\Gr{r_K + {s_K}/{\sqrt{K}}}} \,+\,
\eps_{\bs{b}}} \normtwo{\bs{b}}$$ where $\eps_{\bs{A}}$,
$\eps_{\bs{A}}^{\gr{K}}$,~$\eps_{\bs{b}}$ are defined
in~(\ref{eq:Pert_relative}), $\kappa_{\bs{A}}^{\gr{K}}$,
$\alpha_{\bs{A}}$ in~(\ref{eq:condition_numbers}), and $r_K$, $s_K$
in~(\ref{eq:ratios}). Then the total perturbation obeys
\begin{equation} \label{eq:Total_pert}
\normtwo{\bs{Ex}} + \normtwo{\bs{e}} \;\,\le\;\,
\eps'_{\bs{A},K,\bs{b}}.
\end{equation}
\end{lem}

\begin{proof}
First divide the multiplicative noise term by $\normtwo{\bs{b}}$ and
then apply Lemma~\ref{lem:Bound_Image_Tail}
\begin{eqnarray} \label{eq:second_term}
\frac{\normtwo{\bs{Ex}}}{\normtwo{\bs{Ax}}} &\le&
\frac{\Gr{\normtwo{\bs{E}}^{\gr{K}}\normtwo{\bs{x}_K} +
\normtwo{\bs{E}}\normtwo{\bs{x}_{K^c}}} \cdot
\frac{1}{\sqrt{1-\delta_K}}} {\normtwo{\bs{x}_K} -
\kappa_{\bs{A}}^{\gr{K}}\Gr{\normtwo{\bs{x}_{K^c}} +
\normone{\bs{x}_{K^c}}/{\sqrt{K}}}} \nonumber\\
&=& \frac{\Gr{\normtwo{\bs{E}}^{\gr{K}} +\, \normtwo{\bs{E}}\:\!r_K}
\cdot \frac{1}{\sqrt{1-\delta_K}}} {1 - \kappa_{\bs{A}}^{\gr{K}}
\Gr{r_K + s_K/{\sqrt{K}}}}
\nonumber\\
&\le& \frac{\eps_{\bs{A}}^{\gr{K}}\kappa_{\bs{A}}^{\gr{K}} +\,
\eps_{\bs{A}}\alpha_{\bs{A}}r_K} {1 -
\kappa_{\bs{A}}^{\gr{K}}\Gr{r_K + {s_K}/{\sqrt{K}}}}.
\end{eqnarray}
Including the contribution from the additive noise term completes
the proof.
\end{proof}

%-----------------------------------------------------------------
\subsection{Proof of Theorem \ref{thm:Completely_perturbed_BP}}
\label{sect:Proof_Thm_completely_perturbed_BP}
%-----------------------------------------------------------------
\emph{Step 1.} We duplicate the techniques used in Cand\`{e}s' proof
of Theorem~1.2 in~\cite{CanRIP}, but with decoding matrix~$\bs{A}$
replaced by~$\bs{\hat{A}}$. The proof relies heavily on the RIP for
$\bs{\hat{A}}$ in Theorem~\ref{thm:RIP_perturbed}. Set the BP
minimizer in~(\ref{eq:Complete_Pert_BP}) as $\bs{z^\star} =
\bs{x}+\bs{h}$. Here,~$\bs{h}$ is the perturbation from the true
solution~$\bs{x}$ induced by $\bs{E}$ and~$\bs{e}$. Instead of
Cand\`{e}s'~(9), we now determine that the image of~$\bs{h}$
under~$\bs{\hat{A}}$ is bounded by
\begin{eqnarray}
\normtwo{\bs{\hat{A}h}} &\leq& \normtwo{\bs{\hat{A}z^\star} -
\bs{\hat{b}}} +
\normtwo{\bs{\hat{A}x} - \bs{\hat{b}}} \label{eq:Image of_pert} \\
&\le& 2\,\eps'_{\bs{A},K,\bs{b}}. \nonumber
\end{eqnarray}
The second inequality follows since both terms on the RHS
of~(\ref{eq:Image of_pert}) satisfy the BP constraint
in~(\ref{eq:Complete_Pert_BP}). Notice in the second term
that~$\bs{x}$ is a feasible solution due to
Lemma~\ref{lem:Total_observed_perturbation_bound}.

Since the other steps in the proof are essentially the same, we end
up with constants $\hat{\alpha}$ and $\hat{\rho}$ in Cand\`{e}s'
(14) (instead of~$\alpha$ and $\rho$) where
\begin{equation} \label{eq:perturbed_alpha_rho}
\hat{\alpha} \,:=\,
\frac{2\sqrt{1+\hat{\delta}_{2K}}}{1-\hat{\delta}_{2K}}, \qquad
\hat{\rho} \,:=\,
\frac{\sqrt{2}\,\hat{\delta}_{2K}}{1-\hat{\delta}_{2K}}.
\end{equation}
The final line of the proof concludes that
\begin{equation} \label{eq:Induced_perturbation_h}
\normtwo{\bs{h}} \;\le\;
\frac{2\:\!\hat{\alpha}\:\!{\gr{1+\hat{\rho}}}}{1-\hat{\rho}}
\frac{\normone{\bs{x}-\bs{x}_K}}{\sqrt{K}} \,+\,
\frac{2\:\!\hat{\alpha}}{1-\hat{\rho}}\,\eps'_{\bs{A},K,\bs{b}}.
\end{equation}
The denominator demands that we impose the condition that
$0<1-\hat{\rho}$, or equivalently
\begin{equation} \label{eq:Perturbed_RIC_upper}
\hat{\delta}_{2K} \,<\, \sqrt{2} \:-\:\! 1.
\end{equation}
The constants~$C_0$ and~$C_1$ are obtained by first substituting
$\hat{\alpha}$ and $\hat{\rho}$ from (\ref{eq:perturbed_alpha_rho})
into~(\ref{eq:Induced_perturbation_h}). Then, recalling that
$\hat{\delta}_{2K} \le \hat{\delta}_{2K\!,\:\!\max}$, substitute
$\hat{\delta}_{K\!,\:\!\max}$ from~(\ref{eq:pert_RIC_bound})
(with $K\to2K$).\\[-5pt]

\emph{Step 2.} We still need to show that the hypothesis of
Theorem~\ref{thm:Completely_perturbed_BP} implies
(\ref{eq:Perturbed_RIC_upper}). This is easily verified by
substituting the assumption of $\delta_{2K} < \sqrt{2}
\Gr{1+\eps_{\bs{A}}^{\gr{2K}}}^{-2}-1$
into~(\ref{eq:pert_RIC_bound}) (again with $K\to2K$) and the proof
is complete. \hfill $\blacksquare$

%-----------------------------------------------------------------
\subsection{Proof of Lemma~\ref{lem:Spectral_consequence}}
%-----------------------------------------------------------------
Assume (\ref{cond:Main_thm_constraint_1}) in the hypothesis of
Theorem~\ref{thm:Completely_perturbed_BP}. It is easy to show that
this implies
\begin{equation*}
\normtwo{\bs{E}}^{\gr{2K}} \;<\; \sqrt[4]{2} \:-\:\! \sqrt{1 +
\delta_{2K}}.
\end{equation*}
Simple algebraic manipulation then confirms that
$$\sqrt[4]{2} \:-\:\!
\sqrt{1 + \delta_{2K}} \;<\; \sqrt{1 - \delta_{2K}} \;\le\;
\sigma_{\min}^{\gr{2K}}\gr{\bs{A}}.$$ Therefore,
(\ref{eq:Spectral_consequence}) holds with $k=2K$. Further, for any
$k\le2K$ we have $\sigma_{\max}^{\gr{k}}\gr{\bs{E}} \le
\sigma_{\max}^{\gr{2K}}\gr{\bs{E}}$ and
$\sigma_{\min}^{\gr{2K}}\gr{\bs{A}} \le
\sigma_{\min}^{\gr{k}}\gr{\bs{A}}$, which proves the first part of
the lemma. The second part is an immediate consequence.
\hfill $\blacksquare$\\

%=================================================================
%=================================================================
\section{Classical $\ell_2$ Perturbation Analysis} \label{sect:Classical_Pert_Analysis}
%=================================================================
%=================================================================
Let the subset $T \subseteq \set{1,\ldots,n}$ have cardinality $|T|
= K$, and note the following \emph{$T$-restrictions}:
$\bs{A}_T\in\C^{m\times K}$ denotes the submatrix consisting of the
columns of $\bs{A}$ indexed by the elements of $T$, and similarly
for $\bs{x}_T\in\C^K$.

Suppose the ``oracle'' case where we already know the support~$T$ of
$\bs{x}_K$, i.e., the best $K$-sparse representation
of~$\bs{x}$.\footnote{Although perhaps slightly confusing, note that
$\bs{x}_K\in\C^n$, while $\bs{x}_T\in\C^K$. Restricting $\bs{x}_K$
to its support $T$ yields~$\bs{x}_T$.} By assumption, we are only
interested in the case where $K\le m$ in which $\bs{A}_T$ has full
rank. Given the completely perturbed observation of
(\ref{eq:Pert_observation}), the \emph{least squares problem}
consists of solving:
\begin{equation*}
\bs{z}^\#_T \,=\, \argmin_{\bs{\hat{z}}_T}
\normtwo{\bs{\hat{A}}_T\bs{\hat{z}}_T - \bs{\hat{b}}}.
\end{equation*}
Since we know the support $T$, it is trivial to extend~$\bs{z}^\#_T$
to $\bs{z}^\#\in\C^n$ by zero-padding on the complement of $T$. Our
goal is to see how the perturbations $\bs{E}$ and $\bs{e}$
affect~$\bs{z}^\#$. Using Golub and Van Loan's model
(\!\!\cite{GolubVanLoan}, Thm.~5.3.1) as a guide, assume
\begin{equation} \label{eq:Pert_assumption}
\max\AutoSet{\frac{\normtwo{\bs{E}_T}}{\normtwo{\bs{A}_T}},\,
\frac{\normtwo{\bs{e}}}{\normtwo{\bs{b}}}} \;<\;
\frac{\sigma_{\min}\gr{\bs{A}_T}}{\sigma_{\max}\gr{\bs{A}_T}}.
\end{equation}

\begin{rem}
This assumption is fairly easy to satisfy. In fact,
assumption~(\ref{cond:Main_thm_constraint_1}) in the hypothesis of
Theorem~\ref{thm:Completely_perturbed_BP} immediately implies that
${\normtwo{\bs{E}_T}}/{\normtwo{\bs{A}_T}} <
{\sigma_{\min}\gr{\bs{A}_T}}/{\sigma_{\max}\gr{\bs{A}_T}}$ for all
$\eps_{\bs{A}}^{\gr{2K}} \in [0,\sqrt[4]{2}-1)$. To see this simply
set $k=K$ in~(\ref{eq:Spectral_consequence}) of
Lemma~\ref{lem:Spectral_consequence}, and note that
$\normtwo{\bs{E}_T}\le\normtwo{\bs{E}}^{\gr{K}}$ and
$\sigma_{\min}^{\gr{K}}\gr{\bs{A}}\le\sigma_{\min}\gr{\bs{A}_T}$.
Further, the reasonable condition of $\eps_{\bs{b}} \le
\Gr{\sqrt{2}\Gr{1+\eps_{\bs{A}}^{\gr{2K}}}^2 -1}^{1/2}$ is
sufficient to ensure $\eps_{\bs{b}} <
{\sqrt{1-\delta_{2K}}}/{\sqrt{1+\delta_{2K}}}$ so that
assumption~(\ref{eq:Pert_assumption}) holds. Note that this
assumption has no bearing on CS recovery, nor is it a constraint due
to BP. It is simply made to enable an analysis of the least squares
solution which we use as a best-case comparison below.
\end{rem}

Following the steps in \cite{GolubVanLoan} with the appropriate
modifications for our situation we obtain
\begin{eqnarray}
\normtwo{\bs{z}^\# - \bs{x}_K} &\le& \normtwo{\bs{A}_T^\dagger}
\AutoGroup{\frac{\normtwo{\bs{E}_T\bs{x}_T}}{\normtwo{\bs{Ax}}}
\,+\, \frac{\normtwo{\bs{e}}}{\normtwo{\bs{b}}}}\normtwo{\bs{b}}
\nonumber \\
&\le& \frac{1}{\sqrt{1-\delta_K}}\, \zeta'_{\bs{A},K,\bs{b}}
\nonumber
\end{eqnarray}
where $\bs{A}_T^\dagger = \gr{\bs{A}_T^*\!\bs{A}_T}^{-1}\bs{A}_T^*$
is the left inverse of~$\bs{A}_T$ whose spectral norm
$$\normtwo{\bs{A}_T^\dagger} \,\le\,
\frac{1}{\sqrt{1-\delta_K}},$$
and where
$$\zeta'_{\bs{A},K,\bs{b}} \,:=\, \GRg{
\frac{\kappa_{\bs{A}}^{\gr{K}}\;\!\eps_{\bs{A}}^{\gr{K}}}
{1-\kappa_{\bs{A}}^{\gr{K}}\Gr{r_K + {s_K}/{\sqrt{K}}}}
}\normtwo{\bs{b}}$$
was obtained using the same steps as
in~(\ref{eq:second_term}). Finally, we obtain the total least
squares stability expression
\begin{eqnarray} \label{eq:LS_Pert_relative_error_Compressible}
\normtwo{\bs{z}^\# - \bs{x}} &\le& \normtwo{\bs{x} - \bs{x}_K}
\,+\, \normtwo{\bs{z}^\# - \bs{x}_K}\nonumber\\
&\le& \normtwo{\bs{x} - \bs{x}_K} \,+\,
C_2\:\!\zeta'_{\bs{A},K,\bs{b}},
\end{eqnarray}
with $C_2 = 1/\sqrt{1-\delta_K}$.

%-----------------------------------------------------------------
\subsection{Comparison of LS with BP}
%-----------------------------------------------------------------
Now, we can compare the accuracy of the least squares solution
in~(\ref{eq:LS_Pert_relative_error_Compressible}) with the accuracy
of the BP solution found in~(\ref{eq:Pert_error_soln_CS_l1}).
However, this comparison is not really appropriate when the original
data is compressible since the least squares solution~$\bs{z}^\#$
returns a vector which is strictly $K$-sparse, while the BP
solution~$\bs{z}^\star$ will never be strictly sparse.

To make the comparison fair, we need to assume that $\bs{x}$ is
strictly $K$-sparse. Then, as mentioned previously, the constants
$r_K = s_K =~0$ and the solutions enjoy stability of
\begin{equation*}
\normtwo{\bs{z}^\# - \bs{x}} \;\le\; C_2\:\!
\AutoGroup{\kappa_{\bs{A}}^{\gr{K}}\;\!\eps_{\bs{A}}^{\gr{K}} \,+\,
\eps_{\bs{b}}}\normtwo{\bs{b}},
\end{equation*}
and
\begin{equation*}
{\normtwo{\bs{z^\star} - \bs{x}}} \;\le\; C_1 \AutoGroup{
\kappa_{\bs{A}}^{\gr{K}}\;\!\eps_{\bs{A}}^{\gr{K}} \,+\,
\eps_{\bs{b}}} \normtwo{\bs{b}}.
\end{equation*}

Yet, a detailed numerical comparison of $C_2$ with $C_1$, even at
this point, is still is not entirely valid, nor illuminating. This
is due to the fact that we assumed the oracle setup in the least
squares analysis, which is the best that one could hope for. In this
sense, the least squares solution we examined here can be considered
a ``best, worst-case" scenario. In contrast, the BP solution really
should be thought of as a ``worst, of the worst-case" scenarios.

The important thing to glean is that the accuracy of the BP and the
least squares solutions are both on the order of the noise level
$$\AutoGroup{
\kappa_{\bs{A}}^{\gr{K}}\;\!\eps_{\bs{A}}^{\gr{K}} \,+\,
\eps_{\bs{b}}} \normtwo{\bs{b}}$$ in the perturbed observation. This
is an important finding since, in general, no other recovery
algorithm can do better than the oracle least squares solution.
These results are analogous to the comparison by Cand\`es, Romberg
and Tao in~\cite{CanRomTao_Noise}, although they only consider the
case of additive noise~$\bs{e}$.

%=================================================================
%=================================================================
\section{Conclusion} \label{sect:Conclusion}
%=================================================================
%=================================================================
We introduced a framework to analyze general perturbations in CS and
found the conditions under which BP could stably recover the
original data. This completely perturbed model extends previous work
by including a multiplicative noise term in addition to the usual
additive noise term.

Most of this study assumed no specific knowledge of the
perturbations~$\bs{E}$ and~$\bs{e}$. Instead, the point of view was
in terms of their worst-case relative perturbations $\eps_{\bs{A}},
\eps_{\bs{A}}^{\gr{K}}, \eps_{\bs{b}}$. In real-world applications
these quantities must either be calculated or estimated. This must
be done with care owing to their role in the theorems presented
here.

We derived the RIP for perturbed matrix~$\bs{\hat{A}}$, and showed
that the penalty on the spectrum of its $K$-column submatrices was a
graceful, linear function of the relative perturbation
$\eps_{\bs{A}}^{\gr{K}}$. Our main contribution,
Theorem~\ref{thm:Completely_perturbed_BP}, showed that the stability
of the BP solution of the complectly perturbed scenario was limited
by the total noise in the observation.

Simple numerical examples demonstrated how the multiplicative noise
reduced the accuracy of the recovered BP solution. Formal numerical
simulations were performed on strictly $K$-sparse signals with no
additive noise so as to highlight the effect of
perturbation~$\bs{E}$. These experiments appear to confirm the
conclusion of Theorem~\ref{thm:Completely_perturbed_BP}: the
stability of the BP solution scales linearly with
$\eps_{\bs{A}}^{\gr{K}}$.

We also found that the rank of~$\bs{\hat{A}}$ did not exceed the
rank of~$\bs{A}$ under the assumed conditions. This permitted a
comparison with the oracle least squares solution.

It should be mentioned that designing matrices and checking for
proper RICs is still quite elusive. In fact, the only matrices which
are known to satisfy the RIP (and which have $m~\sim~K$ rows) are
random Gaussian, Bernoulli, and certain partial unitary (e.g.,
Fourier) matrices (see, e.g., \cite{CanTao_NearOptRecovery},
\cite{MPJ06:Uniform}, \cite{RudelsonVershynin_FourGauss}).

%=================================================================
%=================================================================
\appendix{}
\section*{Different cases of perturbation $\bs{E}$}
%=================================================================
%=================================================================
There are essentially two classes of perturbations~$\bs{E}$ which we
care most about: \emph{random} and \emph{structured}. The nature of
these perturbation matrices will have a significant effect on the
value of $\normtwo{\bs{E}}^{\gr{K}}$, which is used in determining
$\eps_{\bs{A}}^{\gr{K}}$ in (\ref{eq:Pert_relative}). In fact,
explicit knowledge of~$\bs{E}$ can significantly improve the
worst-case assumptions presented throughout this paper. However, if
there is no extra knowledge on the nature of~$\bs{E}$, then we can
rely on the ``worst case'' upper bound using the full matrix
spectral norm: $\normtwo{\bs{E}}^{\gr{K}} \le~\normtwo{\bs{E}}.$

%-----------------------------------------------------------------
\subsection{Random Perturbations} \label{sect:Random_Perts}
%-----------------------------------------------------------------
Random matrices, such as Gaussian, Bernoulli, and certain partial
Fourier matrices, are often amenable to analysis with the RIP. For
instance, suppose that $\bs{E}$ is simply a scaled version of a
random matrix $\bs{R}$ so that $\bs{E} = \beta\bs{R}$ with
$0<\beta\ll1$. Denote~$\delta^{\bs{R}}_K$ as the RIC associated with
the matrix $\bs{R}$. Then for all $K$-sparse~$\bs{x}$ the RIP for
matrix $\bs{E}$ asserts
\begin{equation*}
\beta^2 \gr{1-\delta^{\bs{R}}_K}\normtwo{\bs{x}}^2 \,\le\,
\normtwo{\bs{Ex}}^2 \,\le\, \beta^2
\gr{1+\delta^{\bs{R}}_K}\normtwo{\bs{x}}^2,
\end{equation*}
which immediately gives us
\begin{equation*}
\normtwo{\bs{E}}^{\gr{K}} \,\le\, \beta\:\!
\sqrt{1+\delta^{\bs{R}}_K},
\end{equation*}
and thus
\begin{equation*}
\frac{\normtwo{\bs{E}}^{\gr{K}}} {\normtwo{\bs{A}}^{\gr{K}}} \:\le\:
\beta\:\! \frac{\sqrt{1+\delta^{\bs{R}}_K}}{\sqrt{1-\delta_K}}
\:=:\: \eps_{\bs{A}}^{\gr{K}}.
\end{equation*}

%-----------------------------------------------------------------
\subsection{Structured Perturbations} \label{sect:Structured_Perts}
%-----------------------------------------------------------------
Structured matrices (e.g., Toeplitz, banded) are ubiquitous in the
mathematical sciences and engineering. In the CS scenario, suppose
for example that $\bs{E}$ is a partial circulant matrix obtained by
selecting $m$ rows uniformly at random from an $n\times n$ circulant
matrix. An error in the modeling of a communication channel could be
represented by such a partial circulant matrix. When encountering a
structured perturbation such as this it may be possible to exploit
its nature to find a bound $\normtwo{\bs{E}}^{\gr{K}}\le C$.

A complete circulant matrix has the property that each row is simply
a right-shifted version of the row above it. Therefore, knowledge of
any row gives information about the entries of all of the rows. This
is also true for a partial circulant matrix. Thus, with this
information we may be able to find a reasonable upper bound on
$\normtwo{\bs{E}}^{\gr{K}}$. The interested reader can find
relevant literature at~\cite{Rice_CS_resources}.

%=================================================================
%=================================================================
\section*{Acknowledgment}
%=================================================================
%=================================================================
The authors would like to thank Jeffrey Blanchard at the University
of Utah, Deanna Needell and Albert Fannjiang at the University of
California, Davis and the anonymous reviewers. Their comments and
suggestions helped to make the current version of this paper much
stronger.

%=================================================================
%=================================================================
\bibliography{matt_refs}
%=================================================================
%=================================================================

%==========================
%===================================================
%===========================================================================
\end{document}